\documentclass[]{article}
\usepackage[utf8]{inputenc}

\usepackage{mathtools,amsfonts}
\usepackage{csquotes}
\usepackage{cancel}

\usepackage{xspace}
\newcommand*{\FN}{\operatorname{FN}}
\newcommand*{\FP}{\operatorname{FP}}
\newcommand*{\FPR}{\operatorname{FPR}}
\newcommand*{\FNR}{\operatorname{FNR}}
\newcommand*{\DUPLICATE}{\textsf{DUPLICATE}\xspace}
\newcommand*{\UNSEEN}{\textsf{UNSEEN}\xspace}

\newcommand*{\wdup}{\in_w}

\newcommand*{\notwdup}{{{\notin_w}}}

\usepackage{tikz}
\usetikzlibrary{patterns}

\usepackage{pgfplots}
\pgfplotsset{compat=1.14}
\usepgfplotslibrary{groupplots}
\usepgfplotslibrary{fillbetween}
\tikzset{mynode/.style={circle,draw,font=\small,minimum width=1cm,font=\sffamily}}
\pgfdeclarepatternformonly{north east lines wide}%
{\pgfqpoint{-1pt}{-1pt}}%
{\pgfqpoint{10pt}{10pt}}%
{\pgfqpoint{6pt}{6pt}}%
{
	\pgfsetlinewidth{0.1pt}
	\pgfpathmoveto{\pgfqpoint{0pt}{0pt}}
	\pgfpathlineto{\pgfqpoint{6.1pt}{6.1pt}}
	\pgfusepath{stroke}
}

\usepackage[ruled]{algorithm}
\usepackage[noend]{algpseudocode}

\usepackage{xcolor}

\usepackage{varioref}

\usepackage{thmtools}
\usepackage{thm-restate}

\usepackage{multicol}

\usepackage{graphicx}

\usepackage{float}
\restylefloat*{figure}

\usepackage{authblk}
\providecommand{\keywords}[1]
{
	\small	
	\textbf{\textit{Keywords---}} #1
}
\usepackage{amsmath}
\usepackage{amsthm}
\newtheorem{definition}{Definition}
\newtheorem{theorem}{Theorem}

\begin{document}
\title{Approaching Optimal Duplicate Detection in a Sliding Window}
%
%
\author[1]{Rémi Géraud--Stewart}
\author[1,2]{Marius Lombard-Platet}
\author[1]{David Naccache}
\affil[1]{Département d'informatique de l'ENS, ENS, CNRS, PSL Research University, Paris, France}
\affil[2]{Be-ys Research Lab, Clermont-Ferrand, France}
\affil[ ]{\texttt{\{first\}.\{last\}@ens.fr}}
\date{}
%
%
\maketitle              
\begin{abstract}
Duplicate detection is the problem of identifying whether a given item has previously appeared in a (possibly infinite) stream of data, when only a limited amount of memory is available. 


Unfortunately the infinite stream setting is ill-posed, and error rates of duplicate detection filters turn out to be heavily constrained: consequently they appear to provide no advantage, asymptotically, over a biased coin toss \cite{10.1145/3297280.3297335}.


In this paper we formalize the sliding window setting introduced by \cite{She08,Yoo10}, and show that a perfect (zero error) solution can be used up to a maximal window size $w_\text{max}$. Above this threshold we show that some existing duplicate detection filters (designed for the \emph{non-windowed} setting) perform better that those targeting the windowed problem. Finally, we introduce a \enquote{queuing construction} that improves on the performance of some duplicate detection filters in the windowed setting.

We also analyse the security of our filters in an adversarial setting. 
\end{abstract}
\section{Introduction}
\subsection{Motivation}

Throughout this paper, we are interested in the following problem: 
\begin{definition}[Duplicate detection problem over a sliding window, wDDP] Given a stream $E_n = (e_1, e_2, \dotsc, e_n)$, a sliding window size $w$ and a \enquote{new} item $e^\star$, find whether $e^\star$ is also present in the last $w$ elements of the stream, ie., whether $e^\star \in \{e_{n-w+1}, \dotsc, e_n\}$.
At every time increment, the new item is added to the stream, i.e., $E_{n+1} = E_n \mid e^\star$ where $\mid$ denotes concatenation. 
\end{definition}

Note that for $w = \infty$, the problem becomes finding whether an element is a duplicate amongst all previous stream elements. For simplicity in the notation, when we refer to $\infty$DDP we instead write DDP.

Instances of the wDDP abound in computer science, with applications to file system indexation, database queries, network load balancing, network management \cite{10.5555/647912.740658}, in credit card fraud detection \cite{DBLP:journals/corr/abs-1709-08920}, phone calls data mining \cite{10.1145/347090.347094}, etc. A discussion about algorithms on large data streams can be found in \cite{10.1145/776985.776986}. 

In practice, additional constraints exist that we can capture with the following definition:
\begin{definition}[wDDP with bounded memory]
	At every time step $n$, given $e^{\star}$ and a current state (dependent on history) of at most $M$ bits, solve the wDDP for $E_n$ and $e^{\star}$.
\end{definition}
Perfect detection is however not always reachable and it might be more practical to work on a further relaxation of the problem, allowing for errors.


Approximate duplicate detection has many real-life use cases, and can sometimes play a critical role, for instance in cryptographic schemes where all security and secrecy fall apart as soon as a random nonce is used twice, such as the ElGamal \cite{10.1007/3-540-39568-7_2} or ECDSA signatures. Other uses include improvements over caches \cite{4484874}, duplicate clicks \cite{10.1145/1060745.1060753} and others. Please note that approximate detection is a different problem than detection of approximate duplicates \cite{Monge97anefficient}, in which the goal is to find elements similar but not necessarily equal to the target.

On a side note, it is clear that the input distribution plays a central role regarding how efficiently the wDDP can be solved. For instance, some deterministic streams may be expressed very compactly (such as the output of a PRNG with known seed) making the wDDP relatively easy. Information-theoretically, if the source has $U$ bits of entropy then the situation is equivalent to having an $U$-bit, uniformly distributed input. This is the setting we consider here.

As said before, when the window size in wDDP grows infinitely large, it becomes the following problem: find whether $e^\star \in E_n$. 
Unfortunately any solution to this problem will necessary encounter a phenomenon called \enquote{saturation} on large enough data streams \cite{10.1145/3297280.3297335}, and when it happens the algorithm performs no better than at random. 

This is problematic on two grounds: it makes the comparison of several algorithms difficult (since they all asymptotically behave in that fashion), and the unavoidable saturation ruins any particular design's merits. As such, it is more interesting to focus on wDDP rather than DDP.

\subsection{Contributions}
In this paper, we start from a naïve solution for the wDDP to then derive bounds for when it can be solved within $M$ memory bits, up to a window size $w_\text{max}$, in constant time. We then introduce a generalization of the naïve solution, and study its error rate. We show that this construction, which we call Short Hash Filter (SHF), can push the value $w_\text{max}$ further while operating in constant time --- at the cost of some errors. We also provide a different tradeoff, the Compact Short Hash Filter (CSHF), which uses fewer memory but operates in linear time.

Unfortunately, for $w > w_\text{max}$ the performance of SHF degrades very rapidly. We therefore turn our attention to existing data structures designed for the \enquote{non-windowed} setting. We show that some of them outperform dedicated data structures, including SHF, in the $w > w_\text{max}$ regime.

We then introduce the \enquote{queuing construction}, a black box transformation of non-windowed data structures into windowed ones, that improves their performance in the wDDP setting.

Finally, we provide an analysis of our queueing construction's resistance to adversarial streams.

\subsection{Related work}\label{sub:related}
The notion of sliding window was, as far as we know, first introduced in \cite{10.1145/1060745.1060753}; but several variations exist that are incomparable to one another (e.g. \cite{shtul2020agepartitioned}).
The wDDP formulation we rely on is due to \cite{Yoo10,She08}, which also introduce algorithms for solving the wDDP approximately.

The notion of using subfilters, as in the queuing construction, can be found in the A2 filter's design \cite{Yoo10} and a variation thereof can be found in \cite{shtul2020agepartitioned} but in a different DDP formulation. The A2 is built from two Bloom filters, a construction which we generalize and analyse generically in this paper. Similarly, the construction in \cite{shtul2020agepartitioned} only works with Bloom Filters.

A literature review collects the following DDF constructions: A2 filters \cite{Yoo10}, Stable Bloom Filters (SBF) \cite{Den06}, Quotient Hash Tables (QHT) \cite{10.1145/3297280.3297335}, Streaming Quotient Filters (SQF) \cite{Dut13}, Block-decaying Bloom Filters (b\_DBF) \cite{She08}, and a slight variation of Cuckoo Filters \cite{Fan14} suggested by \cite{10.1145/3297280.3297335}. The structure proposed in \cite{10.1145/1060745.1060753} is not designed for wDDP but a variant called `landmark` sliding window, which consists of a zero-resetting of the memory at some user-defined epochs.

\label{sec:dudefisa}
\section{Notations and basic definitions}
We consider an unbounded stream $E = (e_1, e_2, \dotsc, e_n, \dotsc)$ with elements belonging to an alphabet $\Gamma$.

A filter is an algorithm, which has a finite amount of memory $M$ and, for each new element $e$, outputs $\DUPLICATE$ or $\UNSEEN$ whether it thinks $e$ is a duplicate or not.

We usually consider the situations where the available memory is too small for perfect detection, i.e., $M \ll |\Gamma|$. Otherwise, if $M = |\Gamma|$ then the problem can be solved in constant time without errors \cite{10.1145/3297280.3297335}.

An element $e_j$ is a duplicate in $E$ over the sliding window $w$, and we note $e_i \wdup E$ if there exists $j-w \leq i < j$ such that $e_i = e_j$. Otherwise we note $e_i \notwdup E$, and we say $e_i$ is unseen over $w$.
A false positive over $w$ is an element $e\notwdup E$ which is classified as a duplicate, and a false negative is an element $e\wdup E$ which is classified as unseen.

For a filter, the false positive probability ($\FP^w_i$) is the probability that after $i$ insertions, the \emph{unseen} element $e_i\notwdup E$ is a false positive over $w$. The false positive rate $\FPR^w_i$ is the number of false positives divided by the number of unseen elements in E\footnote{We observe that $\FPR^w_n = \frac{1}{n}\sum_{i=1}^n \FP^w_i$, and similarly for $\FNR^w_n$.}. We similarly define the false negative probability $\FN^w_i$, and the false negative rate $\FNR^w_i$.

\paragraph{Remark.} For benchmarking, we usually measure the error rate $ER = \FPR^w + \FNR^w$, as it allows a practical ranking of the solutions. An error rate of $0$ means a perfect filter, while a filter answering randomly has an error rate of $1$. A filter being always wrong has an error rate of $2$.

%

\section{Approximate solution and SHF}

\subsection{Optimal and Approximate Optimal wDDF}
\begin{theorem}\label{thm:opt_wddf}
	For $M \geq w (\log_2(w) + 2\log_2(|\Gamma|))$, the wDDP can be solved exactly (with no errors) in constant time.
\end{theorem}
	
\begin{proof}
	We explicitly construct a DDF that performs the detection. Storing all $w$ elements in the sliding window takes $w\log_2(|\Gamma|)$ memory, using a FIFO queue $Q$; however 
	lookup has a worst-time complexity of $O(w)$. 
	
	We therefore rely on an
	ancillary data structure for the sake of quickly answering lookup questions.
	Namely we use a dictionary $D$ whose keys are elements from $\Gamma$ and values are counters.
	
	When an element $e$ is inserted in the DDF, $e$ is stored and $D[e]$ is incremented (if the key $e$ did not exist in $D$, it is created first, and
	$D[e]$ is set to $1$). In order to keep the number of stored elements to $w$, 
	we discard the oldest element $e_\text{last}$ in $Q$. As we do so, we also decrement $D[e_\text{last}]$, and if $D[e_{\text{last}}] = 0$ the key is deleted from $D$. The whole insertion procedure is therefore performed in constant time.
	
	Lookup of an element $e^\star$ is simply done by looking whether the key $D[e^\star]$ exists, which is done in constant time.
	
	The queue  size is $w\log_2 |\Gamma|$, the dictionary size is $w (\log_2 |\Gamma| + \log_2 w)$ (as the dictionary cannot have more than $w$ keys at the same time, a dictionary key occupies $\log_2 |\Gamma|$ bits and a counter cannot go over $w$, thus being less than $\log_2 w$ bits long). Thus a requirement of $w (\log_2(w) + 2\log_2(|\Gamma|))$ bits for this DDF to work.
	
	Finally this filter does not make any mistake, as the dictionary $D$ keeps an exact account of how many times each element is present in the sliding window.
	\end{proof}

However, this optimal filter requires that the size of $\Gamma$ is known in advance. The dependence on $\log_2 |\Gamma|$ can be dropped, at the cost of allowing errors.

\begin{theorem}\label{thm:SHF}
	Let $w \in \mathbb N$. Let $M \simeq 2w\log_2w$, then the wDDP can be solved with almost no error using $M$ memory bits.
	
	More precisely, it is possible to create a filter of $M$ bits with an FN of $0$, an FP of $1 - (1-\frac1{w^2})^w \sim \frac 1w$, and a time complexity of $O(w)$.
	
	Using $M \simeq 5w\log_2 w$ bits of memory, a constant-time filter with the same error rate can be constructed.
\end{theorem}
Note that we only consider the false positive probability after the filter has inserted at least $w$ elements, i.e., once the filter is full and has reached a stationary regime.

\begin{proof}
	Here again we explicitly construct the filters that attain the theorem's bounds.
	
	Let $h$ be a hash function with codomain $\{0,1\}^{2 \log_2 w}$.
	The birthday theorem \cite{10.1007/3-540-45708-9_19} states that for a hash function $h$ over $a$ bits, one must on average collect $2^{a/2}$ input-output pairs before obtaining a collision. Therefore $2^{(2 \log_2 w) / 2} = w$ hash values $h(e_i)$ can be computed before having a $50\%$ probability of a collision (here, a collision is when two distinct elements of the stream $e_i, e_j$ with $i \neq j, e_i \neq e_j$ have the same hash, i.e. $h(e_i) = h(e_j)$). The 50\% threshold we impose on $h$ is arbitrary but nonetheless practical.
	
	Let $\mathcal F$ be the following DDF: the filter's state consists in a queue of $w$ hashes, and for each new element $e$, $\mathsf{Detect}(e)$ returns \DUPLICATE if $h(e)$ is present in the queue, \UNSEEN otherwise. $\mathsf{Insert}(e)$ appends $h(e)$ to the queue before popping the queue. 
	
	There is no false negative, and a false positive only happens if the new element to be inserted collides with at least one other element, which happens with probability $1 - (1 - \frac 1{2^{2\log_2w}})^{w} = 1 - (1-\frac1{w^2})^w$, hence an FN of $0$ and a FP of $1 - (1-\frac1{w^2})^w$.
	The queue stores $w$ hashes, and as such requires $w \cdot 2\log_2 w$ bits of memory.
	
	Note that this solution has a time complexity of $O(w)$. Using an additional dictionary, as in the previous proof, but with keys of size $2\log_2(w)$, we get a filter with an error rate of about $\frac 1w$ and constant time for insertion and lookup, using $w \cdot 2\log_2 w + w \cdot (2\log_2(w) + \log_2(w)) = 5w\log_2 w$ bits of memory.
	\end{proof}

When $\log_2|\Gamma| > 5 \log_2 w$ this DDF outperforms the naïve strategy\footnote{The naïve strategy consisting of storing the $w$ elements of the sliding window, requiring $w \log_2|\Gamma|$ bits of memory.}, both in terms of time and memory, at the cost of a minimal error. When $\log_2|\Gamma| > 2 \log_2 w$, it outperforms the exact solution described sooner in terms of memory.

\subsection{Short Hash Filter and Compact Short Hash Filter Algorithms}
\paragraph{Short Hash Filter (SHF)}
The approximate filter we described uses hashes of size $2\log_2(w)$ for a given sliding window $w$. However, this hash size is arbitrary, and while the current hash size guarantees a very low error rate, it can be changed. More importantly, in some practical cases the maximal amount of available memory is fixed beforehand. Fixing the memory is also more practical for benchmarking data structures, as it gives the guarantee that all filters operate under the same conditions.

This gives us the Short Hash Filter (SHF), described in Algorithm~\ref{alg:SHF}. The implementation relies on a double-ended queue or a ring buffer, which allows pushing at beginning of a queue and popping at the end in constant time. 

\begin{algorithm}
	\caption{SHF \textsf{Setup}, \textsf{Lookup} and \textsf{Insert}}\label{alg:SHF}
	
	\begin{algorithmic}[1]
	\Function{Setup}{$M, w$}
	\Comment $M$ is the available memory, $w$ the size of the sliding window
	\State $h \gets$ hash function of codomain size $\lfloor \frac M{2w} - \frac 12 \log_2 w \rfloor$
	\State $Q \gets \emptyset$ \Comment $Q$ is a queue of elements of size $h$
	\State $D \gets \emptyset$ \Comment $D$ is a dictionary $h \Rightarrow$  counter (of max value $w$)
	\EndFunction 
	\end{algorithmic}

	\begin{multicols}{2}
	\begin{algorithmic}[1]
		\Function{Insert}{$e$}
		\State $Q.\mathsf{Push\_Front}(h(e))$
		\State $D[h(e)]$++
		
		\If{$Q.\mathsf{length}() > w$}
		\State $h' \gets Q.\mathsf{Pop\_back}()$
		\State $D[h']$-{}-
		\If{$D[h'] = 0$}
		\State Erase key $D[h']$
		\EndIf
		\EndIf
		\EndFunction
	\end{algorithmic}	
		
	\begin{algorithmic}[1]
	\item[]
	\Function{Lookup}{$e$}
	\If{$D[h(e)] > 0$}
	\State \Return \DUPLICATE
	\Else
	\State \Return \UNSEEN
	\EndIf
	\EndFunction
	\item[]
\end{algorithmic}

	\end{multicols}
\end{algorithm}

\paragraph{Compact Short Hash Filter (CSHF)}
Removing the dictionary from the SHF construction yields a more memory-efficient, but less time-efficient construction, which we dub \enquote{compact} short hash filter (CSHF). The CSHF performs in linear time in $w$, and is a simple queue, the only point is that instead of storing $e$, the filter stores $h(e)$, where $h$ is a hash function of codomain size $\lfloor \frac Mw \rfloor$.

%
%
%
%

\paragraph{Error Probabilities.}
Let $w> 0$ be a window size and $M > 0$ the available memory.

We write $\FN^w_\text{SHF}$ the probability of false negative of an SHF with these parameters. We similarly define $\FP^w_\text{SHF}$, $\FN^w_\text{CSHF}$, $\FP^w_\text{CSHF}$. 

\begin{theorem}\label{thm:fprSHF}
	We have:
	\begin{itemize}
		\item $\FN_\text{SHF}^w  = 0$  and $\FP_\text{SHF}^w = 1 - \left(1 - \sqrt{w2^{-M/w}}\right)^w$
		\item$\FN_\text{CSHF}^w  = 0$  and $\FP_\text{CSHF}^w = 1 - \left(1 - 2^{-M/w}\right)^w$
	\end{itemize}	
\end{theorem}

\begin{proof}
	This is an immediate adaptation of the proof from Theorem~\ref{thm:SHF}. An SHF has fingerprints of size $h = \frac M{2w} - \frac 12 \log_2 w$, while a CSHF has fingerprints of size $h' = \frac Mw$.
\end{proof}

\paragraph{Remark:} A CSHF of size $M$ on a sliding window $w$ has the same error rate than an SHF of sliding window $w$ of size $2M + w\log_2 w$.

\paragraph{Saturation.}
SHF has strictly increasing error probabilities, which reach a threshold of $1/2$ for some maximum window size $w_\text{max}$. Beyond this value, these filters saturate extremely quickly: in other words, most SHF will either have an error rate of $0$ or $1$.

An illustration of this phenomenon can be seen in Figure~\ref{fig:shorts}, 
which shows the error rates for SHF with $M=10^5$, against 
a uniformly random stream of $18$-bit elements ($|\Gamma| = 2^{18}$). The benchmark used a finite stream of length $10^6$.

\begin{figure}[t]
	\centering
	\begin{tikzpicture}
\begin{semilogxaxis}[
	legend style={at={(0.02,0.98)},anchor=north west},%
	xlabel=$w$, ylabel=$\FPR^w+\FNR^w (\times 100)$]

\addplot[color=blue, mark=o] table[x=w,y=Error] {graphs/shf.dat};
\addlegendentry{SHF}

\addplot[color=red, mark=x] table[x=a,y=b] {graphs/cshf.dat};
\addlegendentry{CSHF}

\end{semilogxaxis}
\end{tikzpicture}
	\caption{Error rates of SHFs and CSHFs for $M = 10^5$ bits, for varying window sizes $w$.}\label{fig:shorts}
\end{figure}
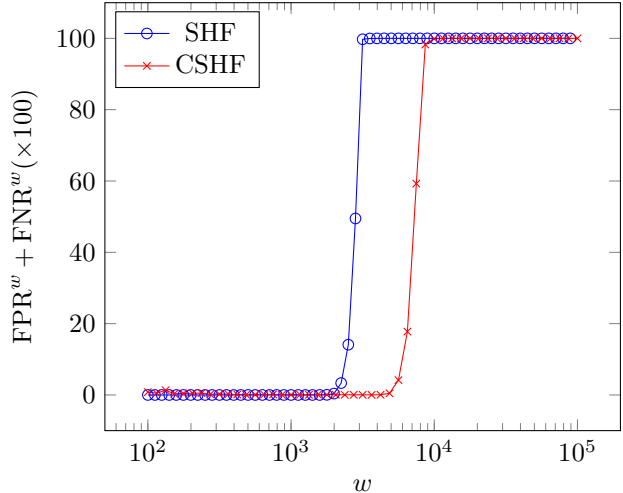

The value $w_\text{max}$ can be obtained by solving (numerically) for $\FP^{w_\text{max}} = 1/2$ for a given $M$. 
Experiments (numerical resolution of $\FP^\text{max} = 1/2$, for about 200 different values of $M$, uniformly distributed \emph{on a log scale} between $10^2$ and $10^6$) indicate an approximately linear relationship between $M$ and $w_\text{max}$: 
$w_\text{max}^\text{CSHF} = 0.0627 M + 443$ ($r^2 = 0.9981$) and
$w_\text{max}^\text{SHF} = 0.0233 M + 186$ ($r^2 = 0.9977$).


\section{Non-windowed DDFs in a wDDP setting}
\subsection{Lower Bound on the Saturation Resistance}
As said in the introduction, it has been proven \cite{10.1145/3297280.3297335} that all filters will reach saturation on the DDP setting. However, they sometimes prove to be efficient in some specific wDDP settings.
This bound is useful for several reasons, notably it provides an estimation of how close to optimality existing filters are.

\begin{theorem}\label{thm:asymptoticSaturation}
	Let $E$ be a stream of $n$ elements uniformly selected from an alphabet of size $|\Gamma|$. For any DDF using $M$ bits of memory, the error probability $EP_n = \FP_n + \FN_n$ satisfies
$$
	EP_n \geq 1 - \frac{1 - \left(1 - \frac{1}{|\Gamma|}\right)^{M}}{1 - \left(1 - \frac 1{|\Gamma|}\right)^n} 
$$
	for any $n > M$.
\end{theorem}
In particular, the asymptotic error rate $EP_\infty$ satisfies 
$$
	EP_\infty \geq \left(1 - \frac{1}{|\Gamma|}\right)^{M} \approx 1 - M/|\Gamma|.
$$

\begin{proof}
	By definition, a perfect filter has the lowest possible error rate.
	With $M$ bits of memory, a perfect filter can store at most $M$ elements in memory \cite[Theorem 2.1]{10.1145/3297280.3297335}. Up to reordering the stream, without loss of generality because it is random, we may assume that the filter stores the $M$ last elements of the stream: any other strategy cannot yield a strictly lower error rate.
	
	If an element is already stored in the filter, then the optimal filter will necessarily answer \DUPLICATE. On the other hand, if the element is not in memory, a perfect filter can choose to answer randomly. Let $p$ be the probability that a filter answers \DUPLICATE when an element is not in memory. An optimal filter will lower the error rate of any filter using the same strategy with a different probability.
	
	An unseen element, by definition, will be unseen in the $M$ last elements of the stream, and hence will not be in the filter's memory, so the filter will return \UNSEEN with probability $1-p$. For this reason, this filter has an FP probability of $p$.
	
	On the other hand, a duplicate $e^\star \in E$ is classified as \UNSEEN if and only if it was not seen in the last $M$ elements of the stream, and the filter answers \UNSEEN. Let $D$ be the event \enquote{\emph{There is at least one duplicate in the stream}} and $C$ be the event \enquote{\emph{There is a duplicate of $e^\star$ in the $M$ previous elements of the stream}}. Then $e^\star$ triggers a false negative with probability
	\begin{align*}
	\FN_n &= (1 - \Pr[C | D])(1-p) \\ &=  \left(1 - \frac{\Pr[C \cap D]}{\Pr[D]}\right)(1-p)\\
	&= \left(1 - \frac{\Pr[C]}{\Pr[D]}\right)(1-p) \\
	&=  \left(1 - \frac{1 - \Pr[\bar C]}{1 - \Pr[\bar D]}\right)(1-p)\\
	\FN_n &=  \left(1 - \frac{1 - \left(1 - \frac 1{|\Gamma|}\right)^M}{1 - \left(1 - \frac 1{|\Gamma|}\right)^n}\right)(1-p)
	\end{align*}
	Hence, the error probability of the filter is
	\begin{align*}
	EP_n 
	& = \FN_n + p \\
	&=  \left(1 - \frac{1 - \left(1 - \frac{1}{|\Gamma|}\right)^{M}}{1 - \left(1 - \frac 1{|\Gamma|}\right)^n}\right)(1-p) + p\\
	EP_n & = 1 - \frac{1-\left(1 - \frac{1}{|\Gamma|}\right)^{M}}{1 - \left(1 - \frac 1{|\Gamma|}\right)^n}(1-p),
	\end{align*}
	which is minimized when $p = 0$. 
	\end{proof}

Note, as highlighted in the proof, that this bound is \emph{not tight}: better bounds may exist, the study of which we leave as an open question for future work.

\subsection{Saturation Resistance of DDFs}\label{sub:saturation}
We now evaluate the saturation rate for several DDFs, in the original DDP setting (without sliding window). Parameters are chosen to yield equivalent memory footprints and were taken from \cite{10.1145/3297280.3297335}, namely:
\begin{itemize}
	\item QHT \cite{10.1145/3297280.3297335}, 1 bucket per row, 3 bits per fingerprint;
	\item SQF \cite{Dut13}, 1 bucket per row, $r = 2$ and $r' = 1$;
	\item Cuckoo Filter \cite{Fan14}, cells containing 1 element of 3 bits each;
	\item Stable Bloom Filter (SBF) \cite{Den06}, 2 bits per cell, 2 hash functions, targeted FPR of 0.02.
\end{itemize}

These filters are run against a stream of uniformly sampled elements from an alphabet of $2^{26}$ elements. This results in around 8\% duplicates amongst the 150 000 000 elements in the longest stream used. 
Results are plotted in Figure~\ref{fig:graph_n}.

\begin{figure}[t]
	\centering
	\begin{tikzpicture}
\begin{axis}[
	xmode=log,
	xmin = 500,
	xmax = 200000000,
	xlabel = {Size of the stream},
	ymax = 110,
    ymin = 0,
    ytick = {0, 20, ..., 120},
    ylabel = {$\FPR^\infty + \FNR^\infty (\times 100)$},
	legend style={at={(0.02,0.98)},anchor=north west,nodes={scale=0.87, transform shape}},%
]
\addplot [color=red, mark=+] coordinates {
	(1000, 0)
	(10000, 0.28)
	(30000, 0.67)
	(50000, 5.55)
	(100000, 8.33)
	(300000, 24.07)
	(1000000, 56.52)
	(3000000, 81.41)
	(10000000, 93.74)
	(50000000, 98.83)
	(100000000, 99.03)
	(150000000, 99.21)
};
\addlegendentry{QHT}

\addplot[color=blue, mark=x] coordinates {
	(1000, 0)
	(10000, 0.49)
	(30000, 1.38)
	(50000, 6.67)
	(100000, 10.4)
	(300000, 33.54)
	(1000000, 64.9)
	(3000000, 85.91)
	(10000000, 95.24)
	(50000000, 98.44)
	(100000000, 99.27)
	(150000000, 99.40)
};
\addlegendentry{SQF}

\addplot [color=teal, mark=triangle] coordinates {
	(1000, 0.2)
	(10000, 2.48)
	(30000, 6.63)
	(50000, 19.81)
	(100000, 28.8)
	(300000, 72.17)
	(1000000, 90.88)
	(3000000, 96.28)
	(10000000, 98.95)
	(50000000, 99.75)
	(100000000, 99.84)
	(150000000, 99.86)
};
\addlegendentry{Cuckoo}

\addplot [color=olive, mark=square, mark options = {scale=0.7}]coordinates {
	(1000, 0.3)
	(10000, 4.07)
	(30000, 10.18)
	(50000, 14.63)
	(100000, 46.64)
	(300000, 77.46)
	(1000000, 92.66)
	(3000000, 97.22)
	(10000000, 99.19)
	(50000000, 99.79)
	(100000000, 99.88)
	(150000000, 99.90)
};
\addlegendentry{SBF}

%
%

\addplot[black, pattern=north east lines wide] coordinates
{
	(1000000, 0)
	(2000000, 49.63)
	(3000000, 66.16)
	(5000000, 79.40)
	(10000000, 89.31)
	(30000000, 95.90)
	(50000000, 97.18)
	(100000000, 98.09)
	(150000000, 98.34)
	(300000000, 98.50)
} \closedcycle; 

\end{axis}
	\end{tikzpicture}
	\caption{Error rate (times 100) of DDFs of 1Mb as a function of stream length. Hatched area represents over-optimal (impossible) values.}\label{fig:graph_n}
\end{figure}
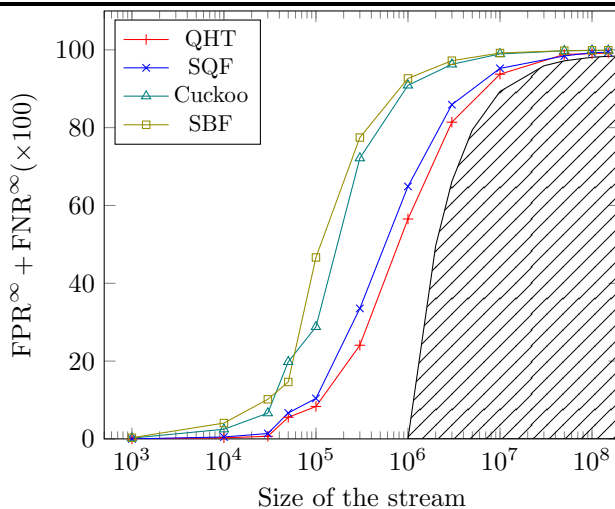

The best results are given by the following filters, in order: QHT, SQF, Cuckoo and SBF. We also observe that QHT and SQF have error rates relatively close to the lower bound, hence suggesting that these filters are close to optimality, especially since the lower bound is not tight.

\subsection{Performance in wDDP}
We now consider the performance of the filters just discussed in the \emph{windowed} setting, for which they were \emph{not} designed. In particular, it is not possible to adjust their parameters as a function of $w$.

Remarkably, some of these filters still outperform dedicated windowed filters for some window sizes at least, as shown in Figure~\ref{fig:nwddf-in-wddp}. In this benchmark, we used the following filters:
\begin{itemize}
	\item block decaying Bloom Filter\footnote{Note that by design, a b\_DBF of $10^5$ bits cannot operate for $w > 6000$.} (b\_DBF) \cite{She08}, sliding window of size $w$
	\item A2 filter \cite{Yoo10}, changing subfilter every $w/2$ insertions
	\item QHT \cite{10.1145/3297280.3297335}, 1 bucket per row, $3$ bits per fingerprint
\end{itemize}

Nevertheless, we will now discuss the queuing construction, which allow us to build windowed filters from the DDP filters.  

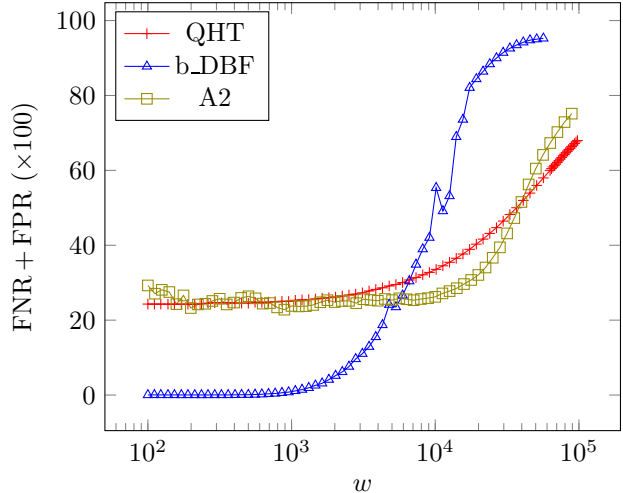
\begin{figure}[t]
	\centering
			\begin{tikzpicture}
	\begin{semilogxaxis}[xlabel = $w$, ylabel = $\FNR + \FPR$ $(\times 100)$, 
	legend style={at={(0.02,0.98)},anchor=north west}]

	\addplot[color=red,mark=+] table[x=w,y=Error] {graphs/b_qht.dat};
\addlegendentry{QHT}

	\addplot[color=blue,mark=triangle] table[x=w,y=Error] {graphs/b_bdbf.dat};
\addlegendentry{b\_DBF}

	\addplot[color=olive,mark=square] table[x=w,y=Error] {graphs/b_a22.dat};
\addlegendentry{A2}
	\end{semilogxaxis}
	\end{tikzpicture}
	\caption{Error rates for QHT, b\_DBF, and A2. While A2 and b\_DBF were designed and adjusted to the wDDP, this is not the case of QHT. Still, QHT outperforms these filters for some values of $w$.}
	\label{fig:nwddf-in-wddp}
\end{figure}

\section{Queuing filters}\label{sec:queue}
We now describe the queuing construction, which produces a sliding window DDF from any DDF. We first give the description of the setup, before studying the theoretical error rates. 
A scheme describing our structure is detailed in Figure~\ref{fig:scheme}.

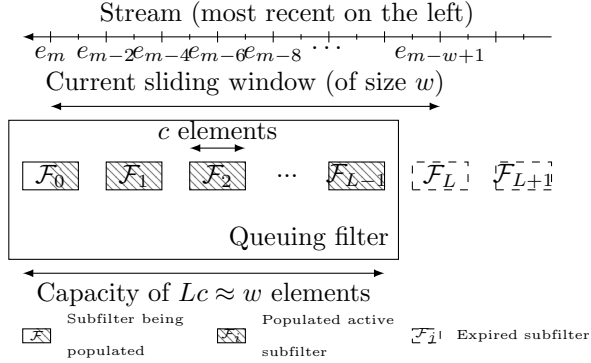
\begin{figure}[t]
	\centering
	\begin{tikzpicture}[scale=0.37]
	
	\draw[pattern=north west lines, pattern color=gray, draw=none] (1, 0) rectangle (2, 1);
	\draw (0, 0) rectangle (2, 1) node[midway] {$\mathcal F_0$};

	\draw[pattern=north west lines, pattern color=gray] (3, 0) rectangle (5, 1) node[midway] {$\mathcal F_1$};
	\draw[pattern=north west lines, pattern color=gray] (6, 0) rectangle (8, 1) node[midway] {$\mathcal F_2$};
	\node at (9.5, 0.5) {...};
	\draw[pattern=north west lines, pattern color=gray] (11, 0) rectangle (13, 1) node[midway] {$\mathcal F_{L-1}$};
	
	\draw[dashed] (14, 0) rectangle (16, 1) node[midway] {$\mathcal F_{L}$};
	\draw[dashed] (17, 0) rectangle (19, 1) node[midway] {$\mathcal F_{L+1}$};
	
	\draw [<->,>=latex] (6, 1.5) -- (8, 1.5) node[above, midway] {$c$ elements};
	\draw [<->,>=latex] (0, -3) -- (13, -3) node[below, midway] {Capacity of $Lc \approx w$ elements};
	\draw [<->,>=latex] (1, 3) -- (15, 3) node[above, midway] {Current sliding window (of size $w$)};
	
	\draw [<-, >=latex] (0, 5.5) -- (19, 5.5) node[above, midway] {Stream (most recent on the left)};
	\foreach \x in {1, 3,...,18}
	\draw (\x, 5.7) -- node[pos=0.5] (point\x) {} (\x, 5.3);
	
	\foreach \x in {2, 4,...,18}
	\draw (\x, 5.6) -- node[pos=0.5] (point\x) {} (\x, 5.4);

	\path (point1) node [below] {$e_m$};
	\foreach \x [evaluate=\x as \i using int(\x-1)] in {3, 5,..., 9}
	\path (point\x) node [below] {$e_{m-\i}$};
	\path (point11) node [below] {$\dots$};
	\path (point15) node [below] {$e_{m - w + 1}$};
	
	\draw (-0.5, 2.5) rectangle (13.5, -2.5) node[anchor=south east] {Queuing filter};
	
	\draw[draw=none, pattern=north west lines, pattern color=gray] (0.5, -5) rectangle (1, -5.5);
	\draw (0, -5) rectangle (1, -5.5) node[midway] {\tiny $\mathcal F$};
	\draw (1.25, -5.25) node[anchor=west,align=left] {\tiny Subfilter being\\ \tiny populated};
	
	\draw[pattern=north west lines, pattern color=gray] (7, -5) rectangle (8, -5.5) node[midway] {\tiny $\mathcal F_i$};
	\draw (8.25, -5.25) node[anchor=west,align=left] {\tiny Populated active\\\tiny subfilter};
	
	\draw[dashed] (14, -5) rectangle (15, -5.5) node[midway] {\tiny $\mathcal F_j$};
	\draw (1 5.25, -5.25) node[anchor=west] {\tiny Expired subfilter};
	\end{tikzpicture}
	\caption{Architecture of the queuing filter, which consists of $L$ subfilters $\mathcal F_i$, each containing up to $c$ elements. Once the newest subfilter has inserted $c$ elements in its structure, the oldest one expires. As such, the latter is dropped and a new one is created and put under population at the beginning of the queue. In this example, the sub-sliding window of $\mathcal F_1$ is $(e_{m-2}, e_{m-3}, e_{m-4})$.} \label{fig:scheme}
\end{figure}

\subsection{The queuing construction}

\paragraph{Principle of operation.}
Let $\mathcal F$ be a DDF. Rather than allocating the whole memory to $\mathcal F$, we will create $L$ copies of $\mathcal F$, each using a fraction of the available memory. Each of these \emph{subfilters} has a limited timespan, and is allowed up to $c$ insertions. The subfilters are organised in a queue.

When inserting a new element in the queuing filter, it is inserted in the topmost subfilter of the queue. After $c$ insertions, a new empty filter is added to the queue, and the oldest subfilter is popped and erased.

As such, we can consider that each subfilter operates on a sub-sliding window of size $c$, which makes the overall construction a DDF operating over a sliding window of size $w = cL$.

\paragraph{Insertion and lookup.}
The filter returns \DUPLICATE if and only if at least one subfilter does. Insertion is a simple insertion in the topmost subfilter.

\paragraph{Queue update.}
After $c$ insertions, the last filter of the queue is dropped, and a new (empty) filter is appended in front of the queue.

\paragraph{Pseudocode.}
We give a brief pseudocode for the queuing filter's functions \textsf{Lookup} and \textsf{Insert}, as well as a \textsf{Setup} function for initialisation, in Algorithm~\ref{alg:queuing}. We introduced for simplicity a constructor $\mathcal F.\textsf{Setup}$ that takes as input an integer $M$ and outputs an initialized empty filter $\mathcal F$ of size at most $M$. Here \texttt{subfilters} is a FIFO that has a \texttt{pop} and \texttt{push\_first} operations, which respectevely removes the last element in the queue or inserts a new item in first position.

\begin{algorithm}[]
	\caption{Queuing Filter \textsf{Setup}, \textsf{Lookup} and \textsf{Insert}}\label{alg:queuing}
	
	\begin{algorithmic}[1]
	\Function{Setup}{$\mathcal F, M, L, c$}
		\Comment $M$ is the available memory, $\mathcal F$ the subfilter structure, $L$ the number of subfilters and $c$ the number of insertions per subfilter
		\State subfilters $\gets \emptyset$
		\State counter $\gets 0$
		\State $m \gets \lfloor M/L \rfloor$
		\For{$i$ from $0$ to $L-1$}
			\State subfilters.push\_first$(\mathcal F.\mathsf{Setup}(m))$
		\EndFor
		\State \textbf{store} (subfilters, $L$, $m$, counter)
	\EndFunction 
	\end{algorithmic}

	\begin{multicols}{2}
	\begin{algorithmic}[1]
	\Function{Lookup}{$e$}
		\For{$i$ from $0$ to $L-1$}
			\If{subfilters[i]$.\mathsf{Lookup(e)}$}
				\State \Return $\DUPLICATE$			
			\EndIf
		\EndFor
		\State \Return \UNSEEN
	\EndFunction
	\item[]
	\end{algorithmic}

	\begin{algorithmic}[1]
	\Function{Insert}{$e$}
		\State subfilters[0]$.\mathsf{Insert}(e)$
		\State counter$++$
		\If{counter == $c$}
			\State subfilters.pop()
			\State subfilters.push\_first($\mathcal F.\textsf{Setup}(m)$)
		\EndIf
	\EndFunction
	\end{algorithmic}
	\end{multicols}
\end{algorithm}

\subsection{Error rate analysis}
The queuing filter's properties can be derived from the subfilters'. False positive and false negative rates are of particular interest. In this section we consider a queuing filter $\mathcal Q$ with $L$ subfilters of type $\mathcal F$ and 
capacity $c$ (which means that the last subfilter is dropped after $c$ insertions).

\paragraph{Remark.}
	\label{subsub:number_elements}
	By definition, after $c$ insertions the last subfilter is dropped.
Information-theoretically, this means that all the information related to the elements inserted in that subfilter has been lost, and there are $c$ such elements by design. 
Therefore, in the steady-state regime, the queuing filter holds information about at least $c(L-1)$  elements (immediately after deleting the last subfilter) and at most $cL$ elements (immediately before this deletion). 

Hence, if $w < cL$, the queuing filter can hold information about \emph{more than $w$ elements}.

\subsubsection{False Positive Probability}

\begin{theorem}
	\label{thm:queue-pfp}
	 Let $\FP_{\mathcal Q, m}^w$ be the false positive probability
	 of $\mathcal Q$ after $m > w$ insertions, over a sliding window of size $w = cL$, we have \\
	$$
		\label{eq:queue-pfp}
		\FP_{\mathcal Q, m}^w = 1 - \left(1 - \FP_{\mathcal F,c}\right)^{L-1} \left(1 - \FP_{\mathcal F, m \bmod c }\right)
	$$ 
	where 
	 $\FP_{\mathcal F,m}$ is the false positive probability of a subfilter $\mathcal F$ after $m$ insertions.
\end{theorem}

\begin{proof}
	Let $E = (e_1, \dotsc, e_m, \dotsc)$ be a stream and $e^\star \notwdup E$. 
	
	Therefore, $e^\star$ is a false positive if and only if at least one subquery $\mathcal F_i.\mathsf{Lookup(e^\star)}$ returns \DUPLICATE. Conversely, $e^\star$ is \emph{not} a false positive when all subqueries $\mathcal F_i.\mathsf{Lookup(e^\star)}$ return \UNSEEN, i.e., when $e^\star$ is not a false positive for each subfilter.
	
	Each subfilter has undergone $c$ insertions, except for the first subfilter which has only undergone $m \bmod c$, we immediately get Eq.~(\ref{eq:queue-pfp}).
\end{proof}

\paragraph{Remark.} In the case $w < cL$, as mentioned previously, there is a non-zero probability that $e^\star$ is in the last subfilter's memory, despite not belonging to the sliding window. 
	
Assuming a uniformly random input stream, and writing $\delta = cL - w$, the probability that $e^\star$ has occurred in $\{e_{m-cL}, \dotsc e_{m-w+1}\}$ is $1-\left(1-\frac1{|\Gamma|}\right)^\delta$. For large $|\Gamma|$ (as is expected to be the case in most applications), this probability is about $\frac{\delta}{|\Gamma|} \ll 1$. Hence, we can neglect the probability that $e^\star$ is present in the filter, and we consider the result of Theorem~\ref{thm:queue-pfp} to be a very good approximation even when $w < cL$.

\subsubsection{False Negative Probability}\label{sub:fnr}
\begin{theorem}
	\label{thm:queue-pfn}
	Let $\FN_{\mathcal Q, m}^w$ be the false negative probability of $\mathcal Q$ after $m > w$ insertions on a sliding window of size $w = cL$, we have
	$$
	\FN_{\mathcal Q, m}^w 
	=  u_c^{L-1} u_{m \bmod c}
	$$
	where 
	we have introduced the short-hand notation $u_\eta= p_{\eta}\FN_{\mathcal F,\eta} + \left(1-p_{\eta}\right)\left(1-\FP_{\mathcal F,\eta}\right)$ where 
	$\FN_{\mathcal F,\eta}$ (resp. $\FP_{\mathcal F, \eta}$) is the false negative probability (resp. false positive) of the subfilter $\mathcal F$ after $\eta$ insertions, and
	$
		p_\eta = \frac{1-\left(1-\frac 1{|\Gamma|}\right)^\eta}{1-\left(1-\frac 1{|\Gamma|}\right)^w} \approx \frac{\eta}{w}.
	$
\end{theorem}

\begin{proof}
	Let $E = (e_1, \dotsc, e_m, \dotsc)$ be a stream, let $w$ be a sliding window and let $e^\star \wdup E$. 
	
	Then $e^\star$ is a false negative if and only if all subfilters $\mathcal F_i$ answer $\mathcal F_i.\mathsf{Detect}(e^\star) = \UNSEEN$. There can be two cases:
	\begin{itemize}
		\item $e^\star$ is present in $\mathcal F_i$'s sub-sliding window;
		\item $e^\star$ is not present in $\mathcal F_i$'s sub-sliding window.
	\end{itemize}
	In the first case, $\mathcal F_i.\mathsf{Detect}(e^\star)$ returns \UNSEEN if and only if $e^\star$ is a false negative for $\mathcal F_i$. This happens with probability $\FN_{\mathcal F,c}$ by definition, except for $\mathcal F_0$, for which the probability is $\FN_{\mathcal F, m\bmod c}$.
	
	In the second case, $\mathcal F_i.\mathsf{Detect}(e^\star)$ returns \UNSEEN if and only if $e^\star$ is not a false positive for $\mathcal F_i$, which happens with probability $1-\FP_{\mathcal F,c}$, execpt for $\mathcal F_0$, for which the probability is $1-\FP_{\mathcal F, m\bmod c}$.
	
	Finally, each event is weighted by the probability $p_c$ that $e^\star$ is in $\mathcal F_i$'s sub-sliding window:
	\begin{align*}
	p_c 
	& = \Pr[\text{$e^\star$ is in }\mathcal F_i \text{ sub-sliding window | } e^\star \wdup E]\\
	& = \frac{\Pr[\text{$e^\star$ is in }\mathcal F_i \text{ sub-sliding window } \cap e^\star \wdup E]}{\Pr[e^\star \wdup E]}\\
	& = \frac{\Pr[\text{$e^\star$ is in }\mathcal F_i \text{ sub-sliding window}]}{\Pr[e^\star \wdup E]}\\
	& = \frac{1 - \Pr[\text{$e^\star$ is not in }\mathcal F_i \text{ sub-sliding window }]}{1 - \Pr[e^\star \notwdup E]}\\
	p_c& = \frac{1 - \left(1 - \frac 1{|\Gamma|}\right)^c}{1 - \left(1-\frac{1}{|\Gamma|}\right)^w}
	\end{align*}
	This concludes the proof.
\end{proof}

\paragraph{Remark.} As previously, the effect of $w < cL$ is negligible for all practical purposes and Theorem~\ref{thm:queue-pfn} is considered a good approximation in that regime. 

\subsection{FNR and FPR}
From the above expressions we can derive relatively compact explicit formulas for the queuing filter's FPR and FNR when $m = cn$ for $n$ a positive integer.
\begin{theorem}
	Let $\FPR_{\mathcal Q, m}^w$ be the false positive rate of $\mathcal Q$ after $m=cn > w$ insertions on a sliding window of size $w = cL$, we have
	\begin{equation*}
	\FPR_{\mathcal Q, cn}^w = 1 -  \frac{(1 - \FP_{\mathcal F, c})^{L-1}}{c}\sum_{\ell = 0}^{c-1} (1-\FP_{\mathcal F, \ell}).
	\end{equation*}
\end{theorem}
\begin{proof}
	
	\begin{align*}
	\FPR_{\mathcal Q, cn}^w
	& = \frac{1}{cn} \sum_{k=1}^{cn} \FP_{\mathcal Q, k}^w 
	= \frac{1}{cn} \sum_{k=1}^{n} \sum_{\ell = 0}^{c-1} \FP_{\mathcal Q, k + \ell}^w 
	= \frac1c \sum_{\ell = 0}^{c-1} \FP_{\mathcal Q, \ell}^w \\
	& = \frac1c \sum_{\ell = 0}^{c-1}  1 - (1 - \FP_{\mathcal F, c})^{L-1}(1 - \FP_{\mathcal F}, \ell) \\
	& = 1 - \frac1c(1 - \FP_{\mathcal F, c})^{L-1}\sum_{\ell = 0}^{c-1} (1-\FP_{\mathcal F, \ell})
	\end{align*}
\end{proof}

\begin{theorem}
	Let $\FNR_{\mathcal Q, m}^w$ be the false negative rate of $\mathcal Q$ after $m=cn > w$ insertions on a sliding window of size $w = cL$, we have
	\begin{equation*}
	\FNR_{\mathcal Q, cn}^w = \frac{u_c^{L-1}}{c}  \sum_{\ell = 0}^{c-1}u_\ell.
	\end{equation*}
\end{theorem}

\begin{proof}
	\begin{align*}
	\FNR_{\mathcal Q, cn}^w 
	& = \frac{1}{cn} \sum_{k=1}^{cn} \FN_{\mathcal Q, k}^w 
	= \frac{1}{cn} \sum_{k=1}^{n} \sum_{\ell = 0}^{c-1} \FN_{\mathcal Q, k + \ell}^w 
	= \frac1c \sum_{\ell = 0}^{c-1} \FN_{\mathcal Q, \ell}^w \\
	& = \frac1c \sum_{\ell = 0}^{c-1} u_c^{L-1} u_{\ell} 
	= \frac{u_c^{L-1}}{c}  \sum_{\ell = 0}^{c-1}u_\ell
	\end{align*}
\end{proof}

As for the probabilities, the expressions derived above for the FNR and FNR are valid to first order in $(w - cL)/|\Gamma|$, i.e. they are good approximations even when $w \approx cL$. 

\subsection{Optimising queuing filters}
\label{sec:optimising}
Let us relax, temporarily, the a priori constraint that $w = cL$. The parameter $L$ determines how many subfilters appear in the queuing construction. Summing up the false positive and false negative rates, we have a total error rate 
$	\operatorname{ER}_{\mathcal Q, cn}^{w}
	= 1 - \alpha \beta^{L-1} + \alpha' \beta'^{L-1}
$,
where $\beta = 1 - \FP_{\mathcal F, c}$, $\beta' = u_c$, 
$\alpha = \frac1c \sum_{\ell = 0}^{c-1}1 - \FP_{\mathcal F, \ell}$ and $\alpha' = \frac1c \sum_{\ell = 0}^{c-1} u_\ell$ depend on $w$, $c$ and the choice of subfilter type $\mathcal F$.

Because $u_\eta = p_{\eta}\FN_{\mathcal F,\eta} + \left(1-p_{\eta}\right)\left(1-\FP_{\mathcal F,\eta}\right)$, differentiating with respect to $L$, knowing that $w = Lc$, and equating the derivative to $0$, one can find the optimal value for $L$ by solving for $x$, which has been obtained via Mathematica:

\begin{align*}
&-\alpha \beta^{-1+x} \log(\beta) 
+ \left(\beta + \FN_{\mathcal F, c} (-1+x)\right)^{-2+x} x^{-x} \left[\vphantom{\left(\FN^{\FN}\right)}\right. 
	-\alpha \beta + \FN_{\mathcal F, c} \left(-\beta (-2+x)+ \FN_{\mathcal F, c} (-1+x)\right) \\
&\quad	+ \left(\alpha+ \FN_{\mathcal F, c} (-1+x)\right)
\times\left(\beta + \FN_{\mathcal F, c} (-1+x)\right) \left(\log\left(\beta + \FN_{\mathcal F, c} (-1+x)\right)-\log(x)\right)
\left.\vphantom{\left(\FN^{\FN}\right)}\right] = 0
\end{align*}

If numerically solving the equation for individual cases is feasible, it seems unlikely that a closed-form formula exists.

\subsection{Queuing filters from existing DDFs}
\label{sec:subfilter-select}
\label{sec:optimal_ddf}
Our queuing construction relies on a choice of subfilters. A first observation 
is that we may assume that all subfilters can be instances of a single 
DDF design (rather than a combination of different designs).

Indeed, a simple symmetry argument shows that a heterogenous selection of subfilters is always worse than a homogeneous one: the crux is that all subfilters play the same role in turn. Therefore we lose nothing by replacing atomically one subfilter by a more efficient one. Applying this to each subfilter we end up with a homogenous selection.

It remains to decide which subfilter construction to choose. The results of an experimental comparison of different DDFs (details about the benchmark are given in Section~\ref{sub:saturation}) are summarized in Figure~\ref{fig:graph_n}. It appears that the most efficient filter (in terms of saturation rate) is the QHT, from \cite{10.1145/3297280.3297335}.



\section{Experiments and Benchmarks}\label{sec:bench}

This section provides details and additional information on the benchmarking experiments run to validate the above analysis. All code is accessible online and will be disclosed after peer review.

\paragraph{Benchmarking queuing filters.}
Applying the queuing construction to DDFs from the literature, we get new filters which are compared in the wDDP setting.

In Section~\ref{sec:optimal_ddf} we suggested the heuristic that the DDFs with the least saturation rate in the DDP would yield the best (error-wise) queuing filter for the wDDP. This heuristic is supported by results, summarized in Figure~\ref{tab:comparing_queuing}. For this benchmark we used the following parameters: uniform stream from an alphabet of size $|\Gamma| = 2^{18}$, memory size $M=100,000$ bits, sliding window of size $w = 10,000$, and we measure the error rate (sum of $\FNR^w$ and $\FPR^w$).

A surprising observation is that when $Lw$ approaches the size of the stream, there is a drop in the error. This is an artifact due to the finite size of our simulations; the stream should be considered infinite, and this drop disappears as the simulation is run for longer (see Appendix~\ref{app:finite}). This effect also alters the error rates for smaller window sizes, albeit much less, and we expect that filter designers care primarily about the small window regime. Nevertheless a complete understanding of this effect would be of theoretical interest, and we leave the study of this phenomenon for future work.

\begin{figure}[t]
		\centering
		\begin{tikzpicture}
\begin{semilogxaxis}[legend cell align = left,
	legend style={at={(0.98,0.02)},anchor=south east, nodes={scale=0.95, transform shape}},%
	ylabel=$\FPR^w + \FNR^w (\times 100)$, xlabel=$w$]
  	\addplot[color=red, mark=+] table[x=w,y=Error] {graphs/qht.dat};
\addlegendentry{Q\_QHT}
  	\addplot[color=blue, mark=x] table[x=w,y=Error] {graphs/sqf.dat};
\addlegendentry{Q\_SQF}

\addplot[color=olive, mark=square, mark options = {scale=0.7}] table[x=w,y=Error] {graphs/sbf.dat};
\addlegendentry{Q\_SBF}

  	\addplot[color=teal, mark=triangle] table[x=w,y=Error] {graphs/cuckoo.dat};
\addlegendentry{Q\_Cuckoo}

\end{semilogxaxis}

\end{tikzpicture}
		\caption{Error rate (times 100) of queuing filters as a function of window size, $M=10^5$, $L=10$, $|\Gamma| = 2^{18}$, on a stream of size $10^7$.}\label{tab:comparing_queuing}
\end{figure}
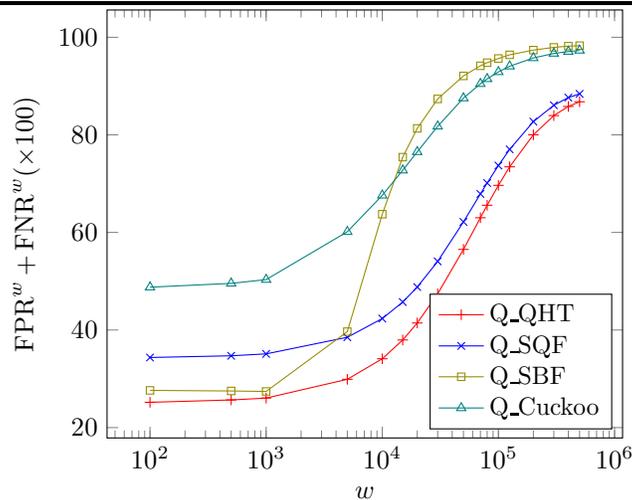


\paragraph{The number of subfilters}\label{sub:L_bench}
The number of subfilters $L$ is an important parameter in the queuing construction, as it affects the filter's error rate in a nontrivial way. An illustration of this dependence is shown in Figure~\ref{fig:bench_l} which plots the error rate of a queueing QHT on an uniform stream of alphabet size $\Gamma = 2^{16}$, with $10^5$ elements in the stream, on various sliding window sizes. 

\begin{figure}[t]
	\centering
		\begin{tikzpicture}
	\begin{semilogxaxis}[legend cell align = right, legend pos = outer north east, ylabel= $\FPR^w + \FNR^w (\times 100)$, xlabel = L]
   \addlegendimage{empty legend}
   \addlegendentry{$w$ value\hspace{.6cm}}
   
	\addplot table[x=L,y=Error] {graphs/1000.dat};
	\addlegendentry{$10^3$}	
	\addplot table[x=L,y=Error] {graphs/5000.dat};
	\addlegendentry{$5\cdot 10^3$}
	\addplot[mark=triangle] table[x=L,y=Error] {graphs/10000.dat};
	\addlegendentry{$10^4$}
	\addplot table[x=L,y=Error] {graphs/30000.dat};
	\addlegendentry{$3\cdot 10^4$}
	\addplot table[x=L,y=Error] {graphs/50000.dat};
	\addlegendentry{$5\cdot 10^4$}
	\addplot table[x=L,y=Error] {graphs/100000.dat};
	\addlegendentry{$10^5$}
	\end{semilogxaxis}
	\end{tikzpicture}
	\caption{Evolution of the error rate of a queueing QHT as a function of $L$, for several window sizes, with $M=10^5$, $|\Gamma| = 2^{18}$, on a stream of size $10^6$.}\label{fig:bench_l}
\end{figure}
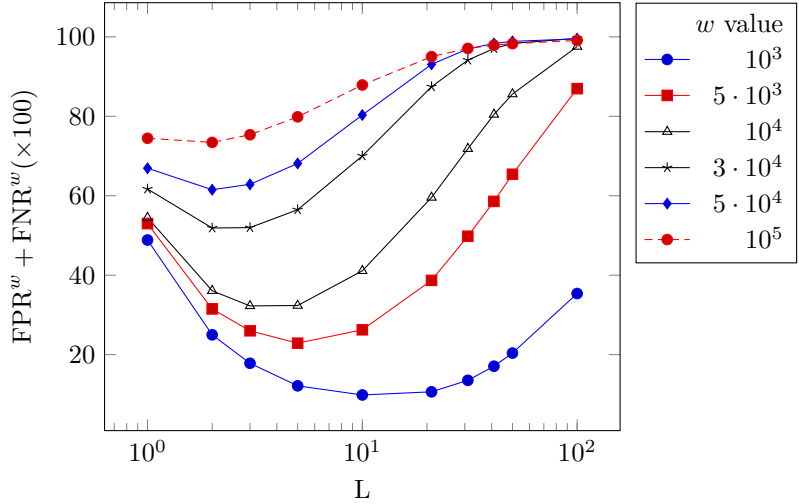
We observe that the optimal value for $L$ does indeed depend on the desired sliding window. However, other experiments on alphabets of other sizes yield very similar results, hence validating the observation made in Section~\ref{sec:optimising} that the optimal number of subfilters does not depends on the alphabet, at least in first approximation.

\paragraph{Filters vs queued filters.}
Using the same stream as previously, we can build queued filters (with an optimal value $L$ for each considered sliding window) and compare their performances to that of non-modified filters. Results on the QHT and SQF are shown in Figure~\ref{fig:improvement}, results for the Cuckoo and SBF are shown in Appendix~\ref{app:improvement}.

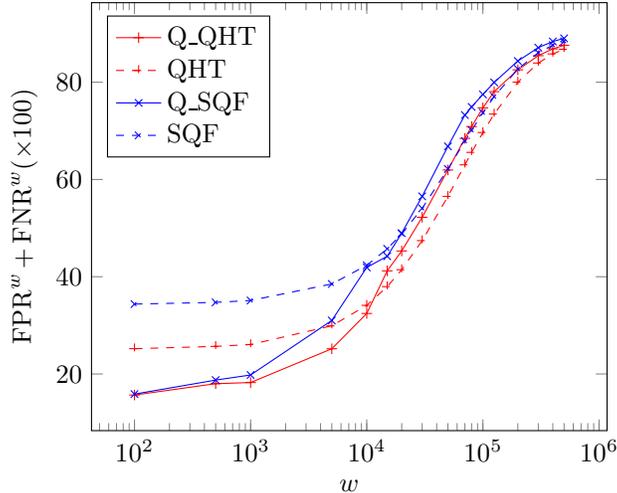
\begin{figure}[t]
	\centering
	\begin{tikzpicture}
\begin{semilogxaxis}[legend cell align = left,%
	 legend pos = north west,
	  xlabel=$w$, ylabel=$\FPR^w + \FNR^w (\times 100)$]
\addplot[color=red, mark=+] table[x=w,y=Error] {graphs/queueing_qht.dat};
\addlegendentry{Q\_QHT};
\addplot[color=red, mark=+, dashed] table[x=w,y=Error] {graphs/qht.dat};
\addlegendentry{QHT};

\addplot[color=blue, mark=x] table[x=w,y=Error] {graphs/queueing_sqf.dat};
\addlegendentry{Q\_SQF};
\addplot[color=blue, mark=x, dashed] table[x=w,y=Error] {graphs/sqf.dat};
\addlegendentry{SQF};

\end{semilogxaxis}

\end{tikzpicture}
	\caption{Comparing performances of QHT and SQF filters, in `vanilla' setting or when placed in our queueing structure.}\label{fig:improvement}
\end{figure}

We observe that queueing filters do not necessarily behave better than their 'vanilla' counterparts, especially on large sliding windows. This can be interpreted by the fact that the DDPs were optimised for infinite sliding windows, and as such operate better than their queueing equivalent on large sliding windows.

\section{Adversarial Resistance of Queueing Filters}
\label{sec:adversarial}

As DDFs have numerous security applications, we now discuss the queuing construction from an adversarial standpoint. We consider an adversarial game in which the attacker wants to trigger false positives or false negatives over the sliding window. One motivation for doing so is causing cache saturation or denial of service by forcing cache misses, triggering false alarms or crafting fradulent transactions without triggering fraud detection systems.

To create a realistic adversary model, we assume like in \cite{DBLP:journals/corr/abs-1709-08920} that the adversary does not have access to the filter's internal memory. Nonetheless, after every insertion she knows whether the inserted element was detected as a duplicate or not. 

We first recall the definition of an adversarial game, adapted to our context.
\begin{definition}
	An adversary $\mathcal A$ feeds data to a sliding window DDF $\mathcal Q$, and for each inserted element, $\mathcal A$ knows whether $\mathcal Q$ answers $\DUPLICATE$ or $\UNSEEN$, but has not access to $\mathcal Q$'s internal state $\mathcal M$. The game has two distinct parts.
\begin{itemize}
	\item In the first part, $\mathcal A$ can feed up to $n$ elements to $\mathcal Q$ and learns $\mathcal Q$'s response for each insertion.
	\item In the second part, $\mathcal A$ sends a unique element $e^\star$.
\end{itemize} 
 $\mathcal A$ wins the $n$-false positive adversarial game (resp. $n$-false negative adversarial game) if and only if $e^\star$ is a false positive (resp. a false negative).
 
 \end{definition}
Variants of these games over a sliding window of size $w$ are immediate.
 

\begin{definition}[Adversarial False Positive Resistance]
	We say that a DDF $\mathcal F$ is $(p, n)$-\emph{resistant to adversarial false positives} if no polynomial-time probabilistic (PPT) adversary $\mathcal A$ can win the $n-$false positive adversarial game with probability greater than $p$.  
\end{definition}
Note that if $\mathcal F$ is $(p,n)$-resistant, then it is $(p, m)$-resistant for all $m < n$.

We define similarly the notion of being \emph{resistant to adversarial false negatives}. Finally, both definitions also make sense in a sliding window context.


\begin{theorem}[Bound on false positive resistance]
	Let $\mathcal Q$ be a filter of $L$ subfilters $\mathcal F_i$, with $c$ insertions maximum per subfilter, let $w$ be a sliding window. 
	
	If $\mathcal F$ is $(p, c)$-resistant to adversarial false positive attacks and $cL \leq w$, then $\mathcal Q$ is $(1-(1-p)^L, w)$-resistant to adversarial false positive attacks on a sliding window of size $w$.
	
	If $cL > w$, the adversary has a success probability of at least $1-(1-p)^L$.
\end{theorem}

\begin{proof}
	If $cL \leq w$, then information-theoretically the subfilters only have information on elements in the sliding window.
	The false positive probability for $\mathcal Q$ is $1 - (1 - \FP_{\mathcal F, c})^L$, which is strictly increasing with $\FP_{\mathcal F,c}$. Hence, the optimal solution is reached by to maximising the false positive probability in each subfilter $\mathcal F_i$. By hypothesis the latter is bounded above by $p$ after $c$ insertions.
	
	On the other hand, if $cL > w$ then the oldest filter holds information about elements that are not in the sliding window anymore. Hence, a strategy for the attacker trying to trigger a false positive on $e^\star$ could be to make it so these oldest elements are all equal to $e^\star$. Let $E$ be the optimal adversarial stream for triggering a false positive on the sliding window $w$ with the element $e^\star$, when $cL \leq w$. The adversary $\mathcal A$ can create a new stream $E' = e^\star | e^\star | \dotsc | E$ where $e^\star$ is prepended $cL - w$ times to $E$.
	
	After $w$ insertions, the last subfilter will answer $\DUPLICATE$ with probability at least $p$, hence giving a lower bound on $\mathcal A's$ success probability. If, for some reason, the last subfilter answers $\DUPLICATE$ with probability less than $p$, then the same reasoning as for when $cL \leq w$ still applies, hence we get the correspondig lower bound (which is, in this case, an equality).
\end{proof}


\begin{theorem}[Bounds on false negative resistance]
	Let $\mathcal Q$ be a filter of $L$ subfilters of kind $\mathcal F$, with $c$ insertions maximum per subfilter, and let $w$ be a sliding window. 
	
	If $\mathcal F$ is $(p, c)$-resistant to adversarial false negative attacks, then $\mathcal A$ can win the adversarial game on the sliding window $w$ with probability at least $p^L$.
	
	Furthermore, for $q$ the lower bound on the false positive probability $\FP_{\mathcal F,c}$ for a given stream, if $w \leq (L-1)c$ then $\mathcal Q$ is $(\min(1 - q, p)^{L-1}p, w)$-resistant to false negative attacks on the sliding window $w$.
	On the other hand if $w > (L-1)c$ then $\mathcal Q$ is  $(\max(1 - q, p)^{L}, w)$-resistant to false negative attacks on the sliding window $w$.
\end{theorem}

\begin{proof}
	Let us first prove that a PPT adversary $\mathcal A$ can win the game with probability at least $p^L$.
	For this, let us consider the adversarial game against the subfilter $\mathcal F$: after $c$ insertions from an aversarial stream $E_c$, $\mathcal A$ choses a duplicate $e^\star$ which will be a false negative with proability $p$.
	Hence, if $\mathcal A$ crafts, for the filter $\mathcal Q$, the following adversarial stream $E' = E_c \mid E_c \mid \cdots \mid E_c$ consisting of $L$ concatenations of the stream $E_c$, then $e^\star$ is a false negative for $\mathcal Q$ if and only if it is a false negative for all subfilters $\mathcal F_i$, hence a success probability for $\mathcal A$ of $p^L$.
	
	Now, Let us prove the case where $w\leq (L-1)c$. In this case, at any time, $\mathcal Q$ remembers all elements from inside the sliding window. As we have seen in the previous example, the success probability of $\mathcal A$ is strictly increasing with the probability of each subfilter to answer $\UNSEEN$. The probability of a subfilter to answer $\UNSEEN$ is:
	
	\begin{itemize}
		\item $\FN'_{\mathcal F,c}$ if $e^\star$ is in the subfilter's sub-sliding window;
		\item $1 - \FP'_{\mathcal F, c}$ if $e^\star$ is not in the subfilter's sub-sliding window
	\end{itemize}
	where $\FN'$ and $\FP'$ are the probabilities of false negative and positives on the adversarial stream (which may be different from a random uniform stream).
	
	However, since $e^\star$ is a duplicate, it is in at least one subfilter's sub-sliding window. As such, the optimal strategy for $\mathcal A$ is to maximise the probability of all subfilters to answer $\UNSEEN$.  
	Now, $\FN'_{\mathcal F,c}$ is bounded above by $p$ and $1 - \FP'_{\mathcal F, c}$ is bounded above by $1-q$, so the best strategy is where as many filters as possible answer $\UNSEEN$ with probability $\max(p, 1-q)$, knowing that at least one filter must contain $e^\star$ and as such its probability for returning $\UNSEEN$ is at most $p$, hence the result.
	
	Now, let us consider the case when $w > (L-1)c$. We have already introduce the element $e^\star$ in the last $w$ elements, and we want to insert it again. It is possible, for the adversary, to create the following stream $E = (e_1, e_2, \dotsc, e_{c-1}, e^\star, e_{c+1}, \dotsc,$ $ e_{Lc}, e_{Lc + 1})$, and to insert $e^\star$ afterwards.
	
	When $e_{Lc+1}$ is inserted, all elements $(e_1, \dotsc, e_{c-1}, e^\star)$ are dropped as the oldest subfilter is popped. Hence, in this context $e^\star$ is not in any subfilter anymore, so by adapting the previous analysis, $\mathcal A$ can get a false negative with probability at most $\max(1 - q, p)^{L}$. 
\end{proof}



%
%
%

 \bibliographystyle{splncs04}

\bibliography{qhtv2}


\begin{thebibliography}{15}


\ifx \showCODEN    \undefined \def \showCODEN     #1{\unskip}     \fi
\ifx \showDOI      \undefined \def \showDOI       #1{#1}\fi
\ifx \showISBNx    \undefined \def \showISBNx     #1{\unskip}     \fi
\ifx \showISBNxiii \undefined \def \showISBNxiii  #1{\unskip}     \fi
\ifx \showISSN     \undefined \def \showISSN      #1{\unskip}     \fi
\ifx \showLCCN     \undefined \def \showLCCN      #1{\unskip}     \fi
\ifx \shownote     \undefined \def \shownote      #1{#1}          \fi
\ifx \showarticletitle \undefined \def \showarticletitle #1{#1}   \fi
\ifx \showURL      \undefined \def \showURL       {\relax}        \fi
\providecommand\bibfield[2]{#2}
\providecommand\bibinfo[2]{#2}
\providecommand\natexlab[1]{#1}
\providecommand\showeprint[2][]{arXiv:#2}

\bibitem[\protect\citeauthoryear{Bloom}{Bloom}{1970}]%
        {bloom_spacetime_1970}
\bibfield{author}{\bibinfo{person}{Burton~H. Bloom}.}
  \bibinfo{year}{1970}\natexlab{}.
\newblock \bibinfo{title}{Space/time trade-offs in hash coding with allowable
  errors}.
\newblock
\newblock
\urldef\tempurl%
\url{https://doi.org/10.1145/362686.362692}
\showURL{%
\tempurl}


\bibitem[\protect\citeauthoryear{Carcillo, Pozzolo, Borgne, Caelen, Mazzer, and
  Bontempi}{Carcillo et~al\mbox{.}}{2017}]%
        {DBLP:journals/corr/abs-1709-08920}
\bibfield{author}{\bibinfo{person}{Fabrizio Carcillo},
  \bibinfo{person}{Andrea~Dal Pozzolo}, \bibinfo{person}{Yann{-}A{\"{e}}l~Le
  Borgne}, \bibinfo{person}{Olivier Caelen}, \bibinfo{person}{Yannis Mazzer},
  {and} \bibinfo{person}{Gianluca Bontempi}.} \bibinfo{year}{2017}\natexlab{}.
\newblock \showarticletitle{{SCARFF:} a Scalable Framework for Streaming Credit
  Card Fraud Detection with Spark}.
\newblock \bibinfo{journal}{\emph{CoRR}}  \bibinfo{volume}{abs/1709.08920}
  (\bibinfo{year}{2017}).
\newblock
\showeprint[arxiv]{1709.08920}
\urldef\tempurl%
\url{http://arxiv.org/abs/1709.08920}
\showURL{%
\tempurl}


\bibitem[\protect\citeauthoryear{Cortes, Fisher, Pregibon, and Rogers}{Cortes
  et~al\mbox{.}}{2000}]%
        {10.1145/347090.347094}
\bibfield{author}{\bibinfo{person}{Corinna Cortes}, \bibinfo{person}{Kathleen
  Fisher}, \bibinfo{person}{Daryl Pregibon}, {and} \bibinfo{person}{Anne
  Rogers}.} \bibinfo{year}{2000}\natexlab{}.
\newblock \showarticletitle{Hancock: A Language for Extracting Signatures from
  Data Streams} \emph{(\bibinfo{series}{KDD ’00})}.
  \bibinfo{publisher}{Association for Computing Machinery},
  \bibinfo{address}{New York, NY, USA}, \bibinfo{pages}{9–17}.
\newblock
\showISBNx{1581132336}
\urldef\tempurl%
\url{https://doi.org/10.1145/347090.347094}
\showDOI{\tempurl}


\bibitem[\protect\citeauthoryear{Demaine, L\'{o}pez-Ortiz, and Munro}{Demaine
  et~al\mbox{.}}{2002}]%
        {10.5555/647912.740658}
\bibfield{author}{\bibinfo{person}{Erik~D. Demaine}, \bibinfo{person}{Alejandro
  L\'{o}pez-Ortiz}, {and} \bibinfo{person}{J.~Ian Munro}.}
  \bibinfo{year}{2002}\natexlab{}.
\newblock \showarticletitle{Frequency Estimation of Internet Packet Streams
  with Limited Space}. In \bibinfo{booktitle}{\emph{Proceedings of the 10th
  Annual European Symposium on Algorithms}} \emph{(\bibinfo{series}{ESA
  ’02})}. \bibinfo{publisher}{Springer-Verlag}, \bibinfo{address}{Berlin,
  Heidelberg}, \bibinfo{pages}{348–360}.
\newblock
\showISBNx{3540441808}


\bibitem[\protect\citeauthoryear{Deng and Rafiei}{Deng and Rafiei}{2006}]%
        {Den06}
\bibfield{author}{\bibinfo{person}{Fan Deng} {and} \bibinfo{person}{Davood
  Rafiei}.} \bibinfo{year}{2006}\natexlab{}.
\newblock \showarticletitle{Approximately detecting duplicates for streaming
  data using stable {Bloom} filters}. In \bibinfo{booktitle}{\emph{{SIGMOD}
  Conference}}. \bibinfo{publisher}{{ACM}}, \bibinfo{pages}{25--36}.
\newblock


\bibitem[\protect\citeauthoryear{Dutta, Narang, and Bera}{Dutta
  et~al\mbox{.}}{2013}]%
        {Dut13}
\bibfield{author}{\bibinfo{person}{Sourav Dutta}, \bibinfo{person}{Ankur
  Narang}, {and} \bibinfo{person}{Suman~K. Bera}.}
  \bibinfo{year}{2013}\natexlab{}.
\newblock \showarticletitle{Streaming Quotient Filter: A Near Optimal
  Approximate Duplicate Detection Approach for Data Streams}.
\newblock \bibinfo{journal}{\emph{Proc. VLDB Endow.}} \bibinfo{volume}{6},
  \bibinfo{number}{8} (\bibinfo{date}{June} \bibinfo{year}{2013}),
  \bibinfo{pages}{589--600}.
\newblock
\showISSN{2150-8097}
\urldef\tempurl%
\url{https://doi.org/10.14778/2536354.2536359}
\showDOI{\tempurl}


\bibitem[\protect\citeauthoryear{ElGamal}{ElGamal}{1985}]%
        {10.1007/3-540-39568-7_2}
\bibfield{author}{\bibinfo{person}{Taher ElGamal}.}
  \bibinfo{year}{1985}\natexlab{}.
\newblock \showarticletitle{A Public Key Cryptosystem and a Signature Scheme
  Based on Discrete Logarithms}. In \bibinfo{booktitle}{\emph{Advances in
  Cryptology}}, \bibfield{editor}{\bibinfo{person}{George~Robert Blakley} {and}
  \bibinfo{person}{David Chaum}} (Eds.). \bibinfo{publisher}{Springer Berlin
  Heidelberg}, \bibinfo{address}{Berlin, Heidelberg}, \bibinfo{pages}{10--18}.
\newblock
\showISBNx{978-3-540-39568-3}


\bibitem[\protect\citeauthoryear{Fan, Andersen, Kaminsky, and Mitzenmacher}{Fan
  et~al\mbox{.}}{2014}]%
        {Fan14}
\bibfield{author}{\bibinfo{person}{Bin Fan}, \bibinfo{person}{Dave~G.
  Andersen}, \bibinfo{person}{Michael Kaminsky}, {and}
  \bibinfo{person}{Michael~D. Mitzenmacher}.} \bibinfo{year}{2014}\natexlab{}.
\newblock \showarticletitle{Cuckoo Filter: Practically Better Than Bloom}
  \emph{(\bibinfo{series}{CoNEXT '14})}. \bibinfo{publisher}{ACM},
  \bibinfo{address}{New York, NY, USA}, \bibinfo{pages}{75--88}.
\newblock
\showISBNx{978-1-4503-3279-8}
\urldef\tempurl%
\url{https://doi.org/10.1145/2674005.2674994}
\showDOI{\tempurl}


\bibitem[\protect\citeauthoryear{G\'{e}raud, Lombard-Platet, and
  Naccache}{G\'{e}raud et~al\mbox{.}}{2019}]%
        {10.1145/3297280.3297335}
\bibfield{author}{\bibinfo{person}{R\'{e}mi G\'{e}raud},
  \bibinfo{person}{Marius Lombard-Platet}, {and} \bibinfo{person}{David
  Naccache}.} \bibinfo{year}{2019}\natexlab{}.
\newblock \showarticletitle{Quotient Hash Tables: Efficiently Detecting
  Duplicates in Streaming Data} \emph{(\bibinfo{series}{SAC ’19})}.
  \bibinfo{publisher}{Association for Computing Machinery},
  \bibinfo{address}{New York, NY, USA}, \bibinfo{pages}{582–589}.
\newblock
\showISBNx{9781450359337}
\urldef\tempurl%
\url{https://doi.org/10.1145/3297280.3297335}
\showDOI{\tempurl}


\bibitem[\protect\citeauthoryear{Golab and \"{O}zsu}{Golab and
  \"{O}zsu}{2003}]%
        {10.1145/776985.776986}
\bibfield{author}{\bibinfo{person}{Lukasz Golab} {and}
  \bibinfo{person}{M.~Tamer \"{O}zsu}.} \bibinfo{year}{2003}\natexlab{}.
\newblock \showarticletitle{Issues in Data Stream Management}.
\newblock \bibinfo{journal}{\emph{SIGMOD Rec.}} \bibinfo{volume}{32},
  \bibinfo{number}{2} (\bibinfo{date}{June} \bibinfo{year}{2003}),
  \bibinfo{pages}{5–14}.
\newblock
\showISSN{0163-5808}
\urldef\tempurl%
\url{https://doi.org/10.1145/776985.776986}
\showDOI{\tempurl}


\bibitem[\protect\citeauthoryear{{Kleanthous} and {Sazeides}}{{Kleanthous} and
  {Sazeides}}{2008}]%
        {4484874}
\bibfield{author}{\bibinfo{person}{M. {Kleanthous}} {and} \bibinfo{person}{Y.
  {Sazeides}}.} \bibinfo{year}{2008}\natexlab{}.
\newblock \showarticletitle{CATCH: A Mechanism for Dynamically Detecting
  Cache-Content-Duplication and its Application to Instruction Caches}. In
  \bibinfo{booktitle}{\emph{2008 Design, Automation and Test in Europe}}.
  \bibinfo{pages}{1426--1431}.
\newblock
\showISSN{1558-1101}
\urldef\tempurl%
\url{https://doi.org/10.1109/DATE.2008.4484874}
\showDOI{\tempurl}


\bibitem[\protect\citeauthoryear{Metwally, Agrawal, and El~Abbadi}{Metwally
  et~al\mbox{.}}{2005}]%
        {10.1145/1060745.1060753}
\bibfield{author}{\bibinfo{person}{Ahmed Metwally}, \bibinfo{person}{Divyakant
  Agrawal}, {and} \bibinfo{person}{Amr El~Abbadi}.}
  \bibinfo{year}{2005}\natexlab{}.
\newblock \showarticletitle{Duplicate Detection in Click Streams}
  \emph{(\bibinfo{series}{WWW ’05})}. \bibinfo{publisher}{Association for
  Computing Machinery}, \bibinfo{address}{New York, NY, USA},
  \bibinfo{pages}{12–21}.
\newblock
\showISBNx{1595930469}
\urldef\tempurl%
\url{https://doi.org/10.1145/1060745.1060753}
\showDOI{\tempurl}


\bibitem[\protect\citeauthoryear{Shen and Zhang}{Shen and Zhang}{2008}]%
        {She08}
\bibfield{author}{\bibinfo{person}{Hong Shen} {and} \bibinfo{person}{Yu
  Zhang}.} \bibinfo{year}{2008}\natexlab{}.
\newblock \showarticletitle{Improved Approximate Detection of Duplicates for
  Data Streams Over Sliding Windows}.
\newblock \bibinfo{journal}{\emph{Journal of Computer Science and Technology}}
  \bibinfo{volume}{23}, \bibinfo{number}{6} (\bibinfo{year}{2008}),
  \bibinfo{pages}{973--987}.
\newblock
\showISSN{1860-4749}
\urldef\tempurl%
\url{https://doi.org/10.1007/s11390-008-9192-1}
\showDOI{\tempurl}


\bibitem[\protect\citeauthoryear{Shtul, Baquero, and Almeida}{Shtul
  et~al\mbox{.}}{2020}]%
        {shtul2020agepartitioned}
\bibfield{author}{\bibinfo{person}{Ariel Shtul}, \bibinfo{person}{Carlos
  Baquero}, {and} \bibinfo{person}{Paulo~Sérgio Almeida}.}
  \bibinfo{year}{2020}\natexlab{}.
\newblock \bibinfo{title}{Age-Partitioned Bloom Filters}.
\newblock
\newblock
\showeprint[arxiv]{cs.DS/2001.03147}


\bibitem[\protect\citeauthoryear{Yoon}{Yoon}{2010}]%
        {Yoo10}
\bibfield{author}{\bibinfo{person}{MyungKeun Yoon}.}
  \bibinfo{year}{2010}\natexlab{}.
\newblock \showarticletitle{Aging Bloom Filter with Two Active Buffers for
  Dynamic Sets}.
\newblock \bibinfo{journal}{\emph{IEEE Trans. on Knowl. and Data Eng.}}
  \bibinfo{volume}{22}, \bibinfo{number}{1} (\bibinfo{date}{Jan.}
  \bibinfo{year}{2010}), \bibinfo{pages}{134--138}.
\newblock
\showISSN{1041-4347}
\urldef\tempurl%
\url{https://doi.org/10.1109/TKDE.2009.136}
\showDOI{\tempurl}


\end{thebibliography}


\begin{thebibliography}{10}
\providecommand{\url}[1]{\texttt{#1}}
\providecommand{\urlprefix}{URL }
\providecommand{\doi}[1]{https://doi.org/#1}

\bibitem{DBLP:journals/corr/abs-1709-08920}
Carcillo, F., Pozzolo, A.D., Borgne, Y.L., Caelen, O., Mazzer, Y., Bontempi,
  G.: {SCARFF:} a scalable framework for streaming credit card fraud detection
  with spark. CoRR  \textbf{abs/1709.08920} (2017),
  \url{http://arxiv.org/abs/1709.08920}

\bibitem{10.1145/347090.347094}
Cortes, C., Fisher, K., Pregibon, D., Rogers, A.: Hancock: A language for
  extracting signatures from data streams. In: Proceedings of the Sixth ACM
  SIGKDD International Conference on Knowledge Discovery and Data Mining. p.
  9–17. KDD ’00, Association for Computing Machinery, New York, NY, USA
  (2000). \doi{10.1145/347090.347094},
  \url{https://doi.org/10.1145/347090.347094}

\bibitem{10.5555/647912.740658}
Demaine, E.D., L\'{o}pez-Ortiz, A., Munro, J.I.: Frequency estimation of
  internet packet streams with limited space. In: Proceedings of the 10th
  Annual European Symposium on Algorithms. p. 348–360. ESA ’02,
  Springer-Verlag, Berlin, Heidelberg (2002)

\bibitem{Den06}
Deng, F., Rafiei, D.: Approximately detecting duplicates for streaming data
  using stable {Bloom} filters. In: {SIGMOD} Conference. pp. 25--36. {ACM}
  (2006)

\bibitem{Dut13}
Dutta, S., Narang, A., Bera, S.K.: Streaming quotient filter: A near optimal
  approximate duplicate detection approach for data streams. Proc. VLDB Endow.
  \textbf{6}(8),  589--600 (Jun 2013). \doi{10.14778/2536354.2536359},
  \url{http://dx.doi.org/10.14778/2536354.2536359}

\bibitem{10.1007/3-540-39568-7_2}
ElGamal, T.: A public key cryptosystem and a signature scheme based on discrete
  logarithms. In: Blakley, G.R., Chaum, D. (eds.) Advances in Cryptology. pp.
  10--18. Springer Berlin Heidelberg, Berlin, Heidelberg (1985)

\bibitem{Fan14}
Fan, B., Andersen, D.G., Kaminsky, M., Mitzenmacher, M.D.: Cuckoo filter:
  Practically better than bloom. In: Proceedings of the 10th ACM International
  on Conference on Emerging Networking Experiments and Technologies. pp.
  75--88. CoNEXT '14, ACM, New York, NY, USA (2014).
  \doi{10.1145/2674005.2674994},
  \url{http://doi.acm.org/10.1145/2674005.2674994}

\bibitem{10.1145/3297280.3297335}
G\'{e}raud, R., Lombard-Platet, M., Naccache, D.: Quotient hash tables:
  Efficiently detecting duplicates in streaming data. In: Proceedings of the
  34th ACM/SIGAPP Symposium on Applied Computing. p. 582–589. SAC ’19,
  Association for Computing Machinery, New York, NY, USA (2019).
  \doi{10.1145/3297280.3297335}, \url{https://doi.org/10.1145/3297280.3297335}

\bibitem{10.1145/776985.776986}
Golab, L., \"{O}zsu, M.T.: Issues in data stream management. SIGMOD Rec.
  \textbf{32}(2),  5–14 (Jun 2003). \doi{10.1145/776985.776986},
  \url{https://doi.org/10.1145/776985.776986}

\bibitem{4484874}
{Kleanthous}, M., {Sazeides}, Y.: Catch: A mechanism for dynamically detecting
  cache-content-duplication and its application to instruction caches. In: 2008
  Design, Automation and Test in Europe. pp. 1426--1431 (March 2008).
  \doi{10.1109/DATE.2008.4484874}

\bibitem{10.1145/1060745.1060753}
Metwally, A., Agrawal, D., El~Abbadi, A.: Duplicate detection in click streams.
  In: Proceedings of the 14th International Conference on World Wide Web. p.
  12–21. WWW ’05, Association for Computing Machinery, New York, NY, USA
  (2005). \doi{10.1145/1060745.1060753},
  \url{https://doi.org/10.1145/1060745.1060753}

\bibitem{Monge97anefficient}
Monge, A., Elkan, C.: An efficient domain-independent algorithm for detecting
  approximately duplicate database records (1997)

\bibitem{She08}
Shen, H., Zhang, Y.: Improved approximate detection of duplicates for data
  streams over sliding windows. Journal of Computer Science and Technology
  \textbf{23}(6),  973--987 (2008). \doi{10.1007/s11390-008-9192-1},
  \url{http://dx.doi.org/10.1007/s11390-008-9192-1}

\bibitem{shtul2020agepartitioned}
Shtul, A., Baquero, C., Almeida, P.S.: Age-partitioned bloom filters (2020)

\bibitem{10.1007/3-540-45708-9_19}
Wagner, D.: A generalized birthday problem. In: Yung, M. (ed.) Advances in
  Cryptology --- CRYPTO 2002. pp. 288--304. Springer Berlin Heidelberg, Berlin,
  Heidelberg (2002)

\bibitem{Yoo10}
Yoon, M.: Aging bloom filter with two active buffers for dynamic sets. IEEE
  Trans. on Knowl. and Data Eng.  \textbf{22}(1),  134--138 (Jan 2010).
  \doi{10.1109/TKDE.2009.136}, \url{https://doi.org/10.1109/TKDE.2009.136}

\end{thebibliography}

\appendix

\section{Effects of the simulation's finiteness}
\label{app:finite}

Theoretical results about the queuing construction apply in principle to an infinite stream. However, simulations are necessarily finite, and for very large windows (that are approximately the same size as the whole stream) this causes interesting artefacts in the error rates.

Note that these effects have very little impact on practical implementations of queuing filters, since almost all use cases assume a window size much smaller than the stream (or, equivalently, a very large stream). Nevertheless we illustrate the effect of the finite simulation and the parameters affecting it, if only to motivate a further analytical study of this phenomenon.

Figure~\ref{fig:strangess} measures the error rate as a function of $w$, for different stream sizes $N$. A visible decrease in $ER$ can be found around $w\approx N$. While we do not have any explanation for the difference in the peaks sizes and exact location, we give the hypothesis that it is related to the choice of $|\Gamma|$.

\begin{figure}[t]
	\centering
		\begin{tikzpicture}
\begin{semilogxaxis}[
	legend columns=2,
	legend pos=south east,
	legend style = {nodes={scale=0.95, transform shape}},
	 xlabel=$w$, ylabel=$\FPR^w + \FNR^w (\times 100)$]
  	\addplot[mark=+,color=blue] table[x=a,y=b] {graphs/10e5.dat};
\addlegendentry{$10^5$};
  	\addplot[mark=x,color=red] table[x=a,y=b] {graphs/10e6.dat};
\addlegendentry{$10^6$}
  	\addplot[mark=triangle,color=olive] table[x=a,y=b] {graphs/10e7.dat};
\addlegendentry{$10^7$}
  	\addplot[mark=square,color=teal] table[x=a,y=b] {graphs/10e8.dat};
\addlegendentry{$10^8$}
\end{semilogxaxis}

\end{tikzpicture}
		\caption{Error rate for queuing QHT ($L = 10$, $M = 10^5$, $|\Gamma| = 2^{16}$) with streams of size $10^5$ to $10^8$.}\label{fig:strangess}
\end{figure}
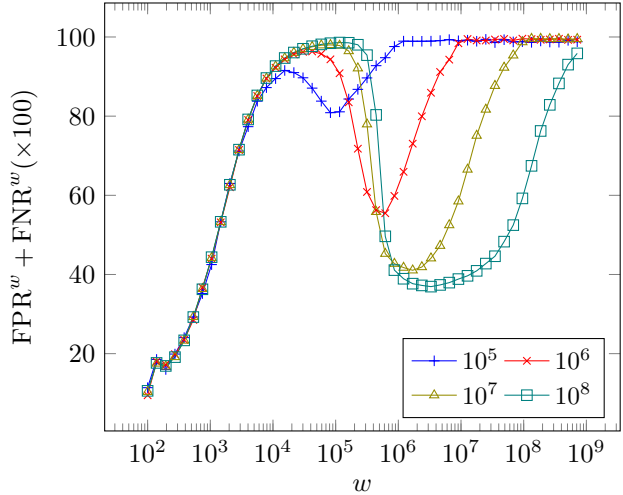

As can be seen on this simulation, there is only disagreement around $w \approx N/L$, and increasing $N$ results in a later and smaller peak. 

It is also possible to run simulations for different alphabet sizes $\Gamma$, which shows that the peak's position increases with $|\Gamma|$, although the relationship is not obvious to quantify.

\section{Filters vs queued filters (complement)}\label{app:improvement}
We here run a comparison of the Cuckoo Filter relative to the Queueing Cuckoo Filter, as well as the SBF relatively to the queueing SBF. The results are given in Figure~\ref{fig:appendix_improvement}.

\begin{figure}[t]
	\centering
	\begin{tikzpicture}
		\begin{semilogxaxis}[
			legend cell align = left,%
			legend pos=south east,
			legend style={nodes={scale=0.7, transform shape}},
			xlabel=$w$, ylabel=$\FPR^w + \FNR^w (\times 100)$]
		\addplot[color=teal, mark=triangle] table[x=w,y=Error] {graphs/queueing_cuckoo.dat};
		\addlegendentry{Q\_Cuckoo};
		\addplot[color=teal, mark=triangle, dashed] table[x=w,y=Error] {graphs/cuckoo.dat};
		\addlegendentry{Cuckoo};
		\addplot[color=olive, mark=square, mark options = {scale=0.7}] table[x=w,y=Error] {graphs/queueing_stablebloom.dat};
		\addlegendentry{Q\_SBF};
		\addplot[color=olive, mark=square, mark options = {scale=0.7}, dashed] table[x=w,y=Error] {graphs/sbf.dat};
		\addlegendentry{SBF};
		\end{semilogxaxis}
	
	\end{tikzpicture}
	\caption{Comparing performances of the Cuckoo and SBF filters, in `vanilla' setting or when placed in our queueing structure.}\label{fig:appendix_improvement}

\end{figure}
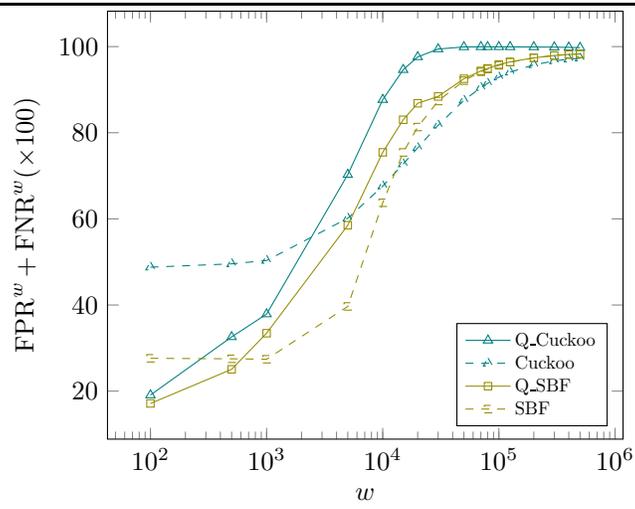

\end{document}